\documentclass[conference]{IEEEtran}

\usepackage{amsmath, amssymb}
\usepackage{algorithm}
\usepackage{algorithmic}
\usepackage{color, array, colortbl}
\usepackage{graphicx}
\usepackage{amsmath, amssymb}
\usepackage{nicefrac}
\usepackage{color, array, colortbl}
\usepackage{multirow}
\usepackage{framed,color}
\usepackage{comment}
\definecolor{shadecolor}{rgb}{0.9,0.9,0.9}
\newcommand{\mathsym}[1]{{}}

\newcommand{\func}[2]{\mathrm{#1}\!\left(#2\right)}
\newcommand{\bigO}[1]{\func{O}{#1}}

\newtheorem{theorem}{Theorem}[section]
\newtheorem{lemma}[theorem]{Lemma}
\newtheorem{claim}[theorem]{Claim}

\newtheorem{definition}[theorem]{Definition}
\newtheorem{corollary}[theorem]{Corollary}

\newcommand{\cancel}[1]{}

\newcommand{\ALG}{\textsc{Sade}}

\newcommand{\E}{\ensuremath{\mathbb{E}}}

\begin{document}

\title{Competitive MAC under Adversarial SINR}

\author {
   Adrian Ogierman$^1$, Andrea Richa$^2$, Christian Scheideler$^1$, Stefan Schmid$^3$, Jin Zhang$^2$\\
   \small $^1$ Department of Computer Science, University of Paderborn, Germany; \{adriano,scheideler\}@upb.de\\
  \small $^2$ Computer Science and Engineering, SCIDSE, Arizona State University, USA; \{aricha,jzhang82@asu.edu\}@asu.edu\\
  \small $^3$ TU Berlin \& Telekom Innovation Labs, Berlin, Germany; stefan@net.t-labs.tu-berlin.de\\
}



\date{}

%

\maketitle

\begin{abstract}
This paper considers the problem of how to efficiently share a wireless medium which is subject to
harsh external interference or even jamming. While this problem has already been studied intensively for
simplistic single-hop or unit disk graph models, we make a leap forward and study MAC protocols
for the SINR interference model (a.k.a.~the \emph{physical model}).

We make two contributions. First, we introduce a new adversarial SINR model which captures a wide range of interference phenomena.
Concretely, we consider a powerful, adaptive adversary which can
\emph{jam} nodes at arbitrary times and which is only limited by some \emph{energy budget}.
The second contribution of this paper is a distributed MAC protocol which provably achieves a constant
competitive throughput in this environment: we show that, with high probability, the protocol ensures that a constant fraction of the non-blocked
time periods is used for successful transmissions.
\end{abstract}


\date{}

\maketitle

\sloppy

\section{Introduction}\label{sec:introduction}

The problem of coordinating the access to a shared medium is a central
challenge in wireless networks. In order to solve this problem, a proper
medium access control (MAC) protocol is needed. Ideally, such a protocol
should not only be able to use the wireless medium as effectively as possible,
but it should also be robust against a wide range of interference problems
including jamming attacks. Currently, the most widely used model to capture
interference problems is the SINR (signal-to-interference-and-noise ratio)
model~\cite{sinr-original}. In this model, a message sent by node $u$ is correctly
received by node $v$ if and only if
$
  P_v(u)/({\cal N}+\sum_{w \in S} P_v(w)) \ge \beta
$
where $P_x(y)$ is the received power at node $x$ of the signal transmitted by
node $y$, ${\cal N}$ is the background noise, and $S$ is the set of nodes
$w\not=u$ that are transmitting at the same time as $u$. The threshold
$\beta>1$ depends on the desired rate, the modulation scheme, etc. When using
the standard model for signal propagation, then this expression results in
$
  P(u)/d(u,v)^{\alpha}/({\cal N}+\sum_{w \in S} P(w)/d(w,v)^{\alpha})
  \ge \beta
$
where $P(x)$ is the strength of the signal transmitted by $x$, $d(x,y)$ is the
Euclidean distance between $x$ and $y$, and $\alpha$ is the path-loss
exponent. In this paper, we will assume that all nodes transmit with some
fixed signal strength $P$ and that $\alpha>2+\epsilon$ for some constant
$\epsilon>0$, which is usually the case in an outdoors environment~\cite{rappaport}.

In most papers on MAC protocols, the background noise ${\cal N}$ is either
ignored (i.e., ${\cal N}=0$) or assumed to behave like a Gaussian variable.
This, however, is an over-simplification of the real world. There are many
sources of interference producing a non-Gaussian noise such as electrical
devices, temporary obstacles, co-existing networks~\cite{podc12stefan},
 or jamming attacks. Also,
these sources can severely degrade the availability of the wireless medium
which can put a significant stress on MAC protocols that have only been
designed to handle interference from the nodes themselves. In order to capture
a very broad range of noise phenomena, one of the main contributions of this
work is the modeling of the background noise ${\cal N}$ (due to jamming or to
environmental noise) with the aid of an adversary ${\cal ADV}(v)$ that has a
fixed energy budget within a certain time frame for each node $v$. More
precisely, in our case, a message transmitted by a node $u$ will be
successfully received by node $v$ if and only if
\begin{equation}
  \frac{P/d(u,v)^{\alpha}}{{\cal ADV}(v)+\sum_{w \in S} P/d(w,v)^{\alpha}}
  \ge \beta
\label{eq:sinr}
\end{equation}
\noindent where ${\cal ADV}(v)$ is the current noise level created by the adversary at node $v$. Our
goal will be to design a MAC protocol that allows the nodes to successfully
transmit messages under this model as long as this is in principle possible.
Prior to our work, no MAC protocol has been shown to have this property.

\cancel{This paper studies the design of efficient medium access protocols
which are \emph{competitive} with respect to an arbitrarily unavailable
medium. For example, these unavailabilities may be due to transmissions in
co-existing networks, due to an external jammer, or any other physical
disturbance in the corresponding frequency domains. In order to capture these
different types of external interference, we use a strong adversarial model:
Our adversary can \emph{adaptively} jam transmissions. Despite this adversary,
we seek to maintain a constant competitive throughput in the arbitrarily
distributed time periods which are non-jammed. Obviously, this is the best we
can hope for.}

\textbf{Model.}
We assume that we have a static set $V$ of $n$ wireless nodes that have arbitrary
fixed positions in the 2-dimensional Euclidean plane so that no two nodes have
the same position. The nodes communicate over a wireless medium
with a single channel. We also assume that the nodes are backlogged in the
sense that they always have something to broadcast. Each node sends at a fixed
transmission power of $P$, and a message sent by $u$ is correctly received by
$v$ if and only if
$  P/d(u,v)^{\alpha}/({\cal ADV}(v)+\sum_{w \in S} P/d(w,v)^{\alpha})
  \ge \beta
$
%
For our formal description and analysis, we assume a synchronized setting
where time proceeds in synchronized time steps called {\em rounds}. In each
round, a node $u$ may either transmit a message or sense the channel, but it
cannot do both. A node which is sensing the channel may either $(i)$ sense an
{\em idle} channel, $(ii)$ sense a {\em busy} channel, or $(iii)$ {\em
receive} a packet. In order to distinguish between an idle and a busy channel,
the nodes use a fixed noise threshold $\vartheta$: if the measured signal power
exceeds $\vartheta$, the channel is considered busy, otherwise idle. Whether a
message is successfully received is determined by the SINR rule described
above.

Physical carrier sensing is part of the 802.11 standard, and is provided by a
Clear Channel Assessment (CCA) circuit. This circuit monitors the environment
to determine when it is clear to transmit. The CCA functionality can be
programmed to be a function of the Receive Signal Strength Indication (RSSI)
and other parameters.  The ability to manipulate the CCA rule allows the MAC
layer to optimize the physical carrier sensing to its needs. Adaptive settings
of the physical carrier sensing threshold have been used, for instance, in
\cite{mobihoc08} to increase spatial reuse.

In addition to the nodes there is an \emph{adversary} that controls the
background noise. In order to cover a broad spectrum of noise phenomena, we
allow this adversary to be adaptive, i.e., for each time step $t$ the
adversary is allowed to know the state of all the nodes in the system at the
beginning of $t$ (i.e., before the nodes perform any actions at time $t$) and
can set the noise level ${\cal ADV}(v)$ based on that for each node $v$. To
leave some chance for the nodes to communicate, we restrict the adversary to
be {\em $(B,T)$-bounded}:
for each node $v$ and time interval $I$ of length $T$, a {\em $(B,T)$-bounded
adversary} has an overall noise budget of $B\cdot T$ that it can use to
increase the noise level at node $v$ and that it can distribute among the time
steps of $I$ as it likes. This adversarial noise model is very general, since in addition to being adaptive, the adversary is allowed to make independent decisions on which nodes to jam at any point in time (provided that the adversary does not exceed its noise budget over a window of size $T$).
 In this way, many noise phenomena can be covered.

Our goal is to design a {\em symmetric local-control} MAC protocol (i.e.,
there is no central authority controlling the nodes, and all the nodes are
executing the same protocol) that has a constant competitive throughput
against any $(B,T)$-bounded adversary as long as certain conditions (on $B$
etc.) are met. In order to define what we mean by ``competitive'', we need
some notation. The {\em transmission range} of a node $v$ is defined as the
disk with center $v$ and radius $r$ with $P/r^{\alpha} \ge \beta \vartheta$.
Given a constant $\epsilon>0$, a time step is called {\em potentially busy} at
some node $v$ if ${\cal ADV}(v) \ge (1-\epsilon)\vartheta$ (i.e., only a
little bit of additional interference by the other nodes is needed so that $v$
sees a busy channel). For a not potentially busy time step, it is still
possible that a message sent by a node $u$ within $v$'s transmission range is
successfully received by $v$. Therefore, as long as the adversary is forced to
offer not potentially busy time steps due to its limited budget and every node
has a least one other node in its transmission range, it is in principle
possible for the nodes to successfully transmit messages. To investigate that
formally, we use the following notation. For any time frame $F$ and node $v$
let $f_v(F)$ be the number of time steps in $F$ that are not potentially busy
at $v$ and let $s_v(F)$ be the number of time steps in which $v$ successfully
receives a message. We call a protocol {\em $c$-competitive} for some time
frame $F$ if
$  \sum_{v \in V} s_v(F) \ge c \sum_{v \in V} f_v(F).
$
An adversary is {\em uniform} if at any time step, ${\cal ADV}(v)={\cal
ADV}(w)$ for all nodes $v,w \in V$, which implies that $f_v(F)=f_w(F)$ for all
nodes. Note that the scope of this paper is not restricted to the case of a uniform jammer (cf~Theorem \ref{thm:main}).

Since the MAC protocol presented in this paper will be randomized, our performance results typically hold \emph{with high probability} (short: \emph{w.h.p.}): this means a probability of at least $1-1/n^c$ for any
constant $c>0$.

\textbf{Our Contribution.}
The contribution of this paper is twofold. First of all, we introduce a novel
extension of the SINR model in order to investigate MAC protocols that are
robust against a broad range of interference phenomena. Second, we present a
MAC protocol called $\ALG$\footnote{\textsc{Sade} stands for SINR
\textsc{Jade}, the SINR variant of the jamming defense protocol
in~\cite{disc10}.} which can achieve a $c$-competitive throughput where $c$
only depends on $\epsilon$ and the path loss exponent $\alpha$ but not on the
size of the network or other network parameters. (In practice, $\alpha$ is typically in the range $2 < \alpha < 5$, and thus $c$ is a constant for fixed $\epsilon$.~\cite{rappaport})
Let $n$ be the number of nodes and
let $N = \max\{n, T \}$. Concretely, we show:

\begin{theorem}\label{thm:main}
When running $\ALG$ for at least $\Omega((T \log N)/\epsilon + (\log
N)^4/(\gamma \epsilon)^2)$ time steps, $\ALG$ has a
$2^{-\bigO{(1/\epsilon)^{2/(\alpha-2)}}}$-competitive throughput for any
$((1-\epsilon)\vartheta,T)$-bounded adversary as long as (a) the adversary is
uniform and the transmission range of every node contains at least one node,
or (b) there are at least $2/\epsilon$ nodes within the transmission range of
every node.
\end{theorem}

On the other hand, we also show the following.

\begin{theorem}\label{thm:main_nonuniform}
The nodes can be positioned so that the transmission range of every node is
non-empty and yet no MAC protocol can achieve any throughput against a
$(B,T)$-bounded adversary with $B > \vartheta$, even if it is 1-uniform.
\end{theorem}

The two theorems demonstrate that our $\ALG$ protocol is basically as robust
as a MAC protocol can get within our model. However, it should be possible to
improve the competitiveness. We conjecture that a polynomial dependency on
$(1/\epsilon)$ is possible, but showing that formally seems to be hard. In
fact, a different protocol than $\ALG$ would be needed for that.

To complement our formal analysis and worst-case bounds, we also report on the results of our simulation study.
This study confirms many of our theoretical results, but also shows that the actual performance is often better
than in the worst-case. For instance, it depends to a lesser extent on $\epsilon$.

\textbf{Paper Organization.}
The remainder of this paper is organized as follows. We present our algorithm
in Section~\ref{sec:algo}, and subsequently analyze its performance in
Section~\ref{sec:analysis}. Simulation results are presented in Section~\ref{sec:simulation}.
After reviewing related work in
Section~\ref{sec:relwork}, we conclude our paper with a discussion in
Section~\ref{sec:conclusion}.

\section{Algorithm}\label{sec:algo}

The intuition behind \ALG\ is simple: Each node $v$ maintains a parameter
$p_v$ which specifies $v$'s probability of accessing the channel at a given
moment of time. That is, in each round, each node $u$ decides to broadcast a
message with probability $p_v$. (This is similar to classical random backoff
mechanisms where the next transmission time $t$ is chosen uniformly at random
from an interval of size $1/p_v$.) The nodes adapt their $p_v$ values over
time in a multiplicative-increase multiplicative-decrease manner, i.e., the
value is lowered in times when the channel is utilized (more specifically, we
decrease $p_v$ whenever a successful transmission occurs) or increased during
times when the channel is idling. However, $p_v$ will never exceed $\hat{p}$,
for some constant $0<\hat{p}<1$ to be specified later.

In addition to the probability value $p_v$, each node $v$ maintains a
time window threshold estimate $T_v$ and a counter $c_v$ for $T_v$. The variable $T_v$ is used to estimate the
adversary's time window $T$: a good estimation of $T$ can help the
nodes recover from a situation where they experience high
interference in the network. In times of high interference, $T_v$
will be increased and the sending probability $p_v$ will be
decreased.


With these intuitions in mind, we can describe \ALG\ in full detail.

\begin{shaded}
Initially, every node $v$ sets $T_v:=1$, $c_v:=1$, and $p_v:=\hat{p}$. In
order to distinguish between idle and busy rounds, each node uses a fixed
noise threshold of $\vartheta$.

The \ALG\ protocol works in synchronized rounds. In every round, each node $v$
decides with probability $p_v$ to send a message. If it decides not to send a
message, it checks the following two conditions:
\begin{itemize}
\item If $v$ successfully receives a message, then $p_v:=(1+\gamma)^{-1}p_v$.
\item If $v$ senses an idle channel (i.e., the total noise created by transmissions of other nodes
and the adversary is less than $\vartheta$), then $p_v:= \min\{(1+\gamma)p_v,
\hat{p}\}, T_v := \max\{ 1, T_v - 1 \}$.
\end{itemize}
Afterwards, $v$ sets $c_v:=c_v+1$. If $c_v>T_v$ then it does the following:
$v$ sets $c_v:=1$, and if there was no idle step among the past $T_v$ rounds,
then $p_v:= (1+\gamma)^{-1} p_v$ and $T_v:=T_v+2$.
\end{shaded}

In order for \ALG\ to be constant competitive in terms of throughput, the
parameter $\gamma$ needs to be a sufficiently small value that depends very
loosely on $n$ and $T$. Concretely, $\gamma\in O(1/(\log T + \log\log n))$.

Our protocol \ALG\ is an adaption of the MAC protocol described
in~\cite{disc10} for Unit Disk Graphs that works in more realistic network
scenarios considering physical interference. The main difference in the new
protocol is that in order to use the concepts of idle and busy rounds, the
nodes employ a fixed noise threshold $\vartheta$ to distinguish between idle (noise
$< \vartheta$) and busy rounds (noise $\ge \vartheta$): in some scenarios the threshold may not be representative, in the sense that,
since the
success of a transmission depends on the noise at the receiving node and on
$\beta$, it can happen that a node senses an idle or busy channel while
\emph{simultaneously} successfully receiving a message. In order to deal with
this problem, $\ALG$ first checks whether a message is successfully received,
and \emph{only otherwise} takes into account whether a channel is idle or busy.
Another change to the protocol in~\cite{disc10} is that we adapt $T_v$ based
on idle time steps which allows us to avoid the upper bound on $T_v$ in the
protocol in \cite{disc10} so that our protocol is more flexible.

\section{Analysis}\label{sec:analysis}

While the MAC protocol \ALG\ is very simple, its stochastic analysis is rather
involved: it requires an understanding of the complex interplay of the nodes
following their randomized protocol in a dependent manner. In particular, the
nodes' interactions  depend on their distances (the geometric setting). In
order to study the throughput achieved by \ALG, we will consider some fixed
node $v\in V$ and will divide the area around $v$ into three circular and
concentric \emph{zones}.

Let $D_R(v)$ denote the \emph{disk} of radius $R$ around a given node $v\in
V$. In the following, we will sometimes think of $D_R(v)$ as the corresponding
geometric area on the plane, but we will also denote by $D_R(v)$ the \emph{set
of nodes} located in this area. The exact meaning will be clear from the
context. Moreover, whenever we omit $R$ we will assume $R_1$ as radius, where
$R_1$ is defined as in Definition~\ref{def:zones}.

\begin{definition}[Zones]\label{def:zones}
Given any node $v\in V$, our analysis considers three zones around $v$,
henceforth referred to as Zone~1, Zone~2,
and Zone~3: Zone~1 is the disk of radius $R_1$ around $v$,
Zone~2 is the disk of radius $R_2$ around $v$ \emph{minus} Zone~1,  and
Zone~3 is the remaining part of the plane. Concretely:
\begin{enumerate}
\item Zone 1 covers the transmission range of $v$, i.e., its radius
$R_1$ is chosen so that $P/R_1^{\alpha} \ge \beta \vartheta$, which implies
that $R_1 = \sqrt[\alpha]{P/(\beta \vartheta)}$. Region $D_{R_1}(v)$ has the
property that if there is at least one sender $u\in D_{R_1}(v)$, then $v$ will
either successfully receive the message from $u$ or sense a busy channel, and
$v$ will certainly receive the message from $u$ if the overall interference caused by
other nodes and the adversary is at most $\vartheta$.

\item Zone 2 covers a range that we call the
{\em (critical) interference range} of $v$. Its radius $R_2$ is chosen in a
way so that if none of the nodes in Zone 1 and Zone 2 transmit a message, then
the interference at any node $w \in D_{R_1}(v)$ caused by transmitting nodes
in Zone 3 is likely to be less than $\epsilon \vartheta$. Hence, if the
current time step is potentially non-busy at some $w \in D_{R_1}(v)$ (i.e.,
${\cal ADV}(w) \le (1-\epsilon)\vartheta$), then the overall inference at $w$
is less than $\vartheta$, which means that $w$ will see an idle time step. It
will turn out that $R_2$ can be chosen as $\bigO{(1/\epsilon)^{1/(\alpha-1)}
R_1}$.

\item Everything outside of Zone~2 is called \emph{Zone~3}.
\end{enumerate}
Whenever it is clear from the context, we use $D_1, D_2$, and $D_3$ instead of
$D_{R_1}, D_{R_2}$, and the area covered by Zone 3, respectively.
\end{definition}


The key to proving a constant competitive throughput is the analysis of the
aggregate probability (i.e., the sum of the individual sending probabilities
$p_v$) of nodes in disks $D_{1}(v)$ and $D_{2}(v)$: We will show that the
expected aggregate probabilities of $D_{1}(v)$ and $D_{2}(v)$, henceforth
referred to by $p_1$ and $p_2$, are likely to be at most a constant. Moreover,
our analysis shows that while the aggregate probability $p_3$ of the
potentially infinitely large Zone~3 may certainly be unbounded (i.e., grow as
a function of $n$), the aggregated power received at any node $w \in
D_{1}(v)$ from all nodes in Zone~3 is also constant on expectation.

\subsection{Zone 1}
To show an upper bound on
$p_1=\sum_{u\in D_1(v)} p_u$, i.e., the aggregate probability of the nodes in
Zone~1 of $v$, we can follow a strategy similar to the one introduced for the
Unit Disk Graph protocol~\cite{disc10}.

In the following, we assume that the budget $B$ of the adversary is limited by
$(1-\epsilon')\vartheta$ for some constant $\epsilon'=2\epsilon$. In this
case, $B$ is at most $(1-\epsilon)^2 \vartheta$. We first look at a slightly
weaker form of adversary. We say that a round $t$ is {\em open} for a node $v$
if $v$ and at least one other node $w$ within its transmission range are
potentially non-busy, i.e., ${\cal ADV}(v) \le (1-\epsilon)\vartheta$ and
${\cal ADV}(w) \le (1-\epsilon)\vartheta$ (which also implies that $v$ has at
least one node within its transmission range). An adversary is {\em weakly
$(B,T)$-bounded} if it is $(B,T)$-bounded and in addition to this, at least a
constant fraction of the potentially non-busy rounds at each node is open in
every time interval of size $T$. We will show the following result:

\begin{theorem} \label{th:main}
When running $\ALG$ for at least $\Omega((T \log N)/\epsilon' + (\log
N)^4/(\gamma \epsilon')^2)$ time steps, $\ALG$ has a
$2^{-\bigO{(1/\epsilon')^{2/(\alpha-2)}}}$-competitive throughput for any
weakly $((1-\epsilon')\vartheta,T)$-bounded adversary.
\end{theorem}

In order to prove this theorem, we focus on a {\em time frame} $I$ of size $F$
consisting of $\delta \log N / \epsilon$ {\em subframes} $I'$ of size
$f=\delta[T+(\log^3 N)/(\gamma^2 \epsilon)]$ each, where $f$ is a multiple of
$T$, $\delta$ is a sufficiently large constant, and $N = \max\{T, n\}$.
Consider some fixed node $v$. We partition $D_1(v)$ into six \emph{sectors} of
equal angles from $v$, $S_1,...,S_6$. Note that for any sector $S_i$ it holds
that if a node $u \in S_i$ transmits a message, then its signal strength at
any other node $u' \in S_i$ is at least $\beta \vartheta$. Fix a sector $S$
and consider some fixed time frame $F$. Let us refer to the sum of the sending
probabilities of the neighboring nodes of a given node $v\in S$ by $\bar{p}_v
:= \sum_{w\in S \setminus \{v\}} p_w$. The following lemma, which is proven in
\cite{disc10}, shows that $p_v$ will decrease dramatically if $\bar{p}_v$ is
high throughout a certain time interval.

\begin{lemma}\label{single_node}
Consider any node $w$ in $S$. If $\bar{p}_w>5-\hat{p}$ during \emph{all}
rounds of a subframe $I'$ of $I$ and at the beginning of $I'$, $T_w \le
\sqrt{F}$, then $p_w$ will be at most $1/n^2$ at the end of $I'$, w.h.p.
\end{lemma}

Given this property of the individual probabilities, we can derive an upper
bound for the aggregate probability of a Sector $S$. In order to compute
$p_S=\sum_{v\in S}p_v$, we introduce three thresholds, a low one,
$\rho_{green}=5$, one in the middle, $\rho_{yellow}=5e$, and a high one,
$\rho_{red}=5e^2$. The following three lemmas provide some important insights
about these probabilities. The first lemma is shown in \cite{disc10}.

\begin{lemma}\label{fast:recovery}
Consider any subframe $I'$ in $I$. If at the beginning of $I'$, $T_w \le
\sqrt{F}$ for all $w \in S$, then there is at least one round in $I'$ with
$p_S \le \rho_{green}$ w.h.p.
\end{lemma}

\begin{lemma}\label{lemma:helper}
For any subframe $I'$ in $I$ it holds that if $p_S \le \rho_{green}$ at the
beginning of $I'$, then $p_S \le \rho_{yellow}$ throughout $I'$,
w.m.p.\footnote{With moderate probability, or w.m.p., means a probability of
at least $1-\log^{-\Omega(1)}n$.} Similarly, if $p_S \le \rho_{yellow}$ at the
beginning of $I'$, then $p_S \le \rho_{red}$ throughout $I'$, w.m.p. The
probability bounds hold irrespective of the events outside of $S$.
\end{lemma}
\begin{proof}
It suffices to prove the lemma for the case that initially $p_S \le
\rho_{green}$ as the other case is analogous. Consider some fixed round $t$ in
$I'$. Let $p_S$ be the aggregate probability at the beginning of $t$ and
$p'_S$ be the aggregate probability at the end of $t$. Moreover, let
$p_{S}^{(0)}$ denote the aggregate probability of the nodes $w \in S$ with a total interference of less than $\vartheta$ in round $t$ when ignoring the
nodes in $S$. Similarly, let $p_{S}^{(1)}$ denote the aggregate probability
of the nodes $w \in S$ with a single transmitting node in $D_1(w) \setminus S$
and additionally an interference of less than $\vartheta$ in round $t$, and let
$p_{S}^{(2)}$ be the aggregate probability of the nodes $w \in S$ that do not
satisfy the first two cases (which implies that they will not experience an
idle channel, no matter what the nodes in $S$ will do). Certainly, $p_S =
p_S^{(0)}+p_S^{(1)}+p_S^{(2)}$. Our goal is to determine $p'_S$ in this case.
Let $q_0(S)$ be the probability that all nodes in $S$ stay silent, $q_1(S)$ be
the probability that exactly one node in $S$ is transmitting, and $q_2(S) =
1-q_0(S)-q_1(S)$ be the probability that at least two nodes in $S$ are
transmitting.

First, let us ignore the case that $c_v > T_v$ for a node $v \in S$ at round
$t$. By distinguishing 9 different cases, we obtain the following result:
$
\E[p'_S] \leq q_0(S) \cdot $ $[(1+\gamma) p_S^{(0)} + (1+\gamma)^{-1} p_S^{(1)} + p_S^{(2)}]$
        $ + q_1(S) \cdot [(1+\gamma)^{-1} p_S^{(0)} + p_S^{(1)}+p_S^{(2)}] $ $
        + q_2(S) \cdot [p_S^{(0)} + p_S^{(1)} + p_S^{(2)}] $
Just as an example, consider the case of $q_0(S)$ and $p_S^{(1)}$, i.e., all
nodes in $S$ are silent and for all nodes in $w \in S$ accounted for in
$p_S^{(1)}$ there is exactly one transmitting node in $D_1(w) \setminus S$ and
the remaining interference is less than $\vartheta$. In this case, $w$ is guaranteed
to receive a message, so according to the $\ALG$ protocol, it lowers $p_w$ by
$(1+\gamma)$.

The upper bound on $\E[p'_S]$ certainly also holds if $c_v > T_v$ for a node
$v \in S$ because $p_v$ will never be increased (but possibly decreased) in
this case.
For the rest of the proof we refer the reader to \cite{disc10}.
\end{proof}

\begin{lemma}\label{sector}
For any subframe $I'$ in $I$ it holds that if there has been at least one
round during the past subframe where $p_S \le \rho_{green}$, then throughout
$I'$, $p_S \le \rho_{red}$ w.m.p., and the probability bound holds
irrespective of the events outside of $S$.
\end{lemma}
\begin{proof}
Suppose that there has been at least one round during the past subframe where
$p_S \le \rho_{green}$. Then we know from Lemma~\ref{lemma:helper} that w.m.p.
$p_S \le \rho_{yellow}$ at the beginning of $I'$. But if $p_S \le
\rho_{yellow}$ at the beginning of $I'$, we also know from
Lemma~\ref{lemma:helper} that w.m.p. $p_S \le \rho_{red}$ throughout $I'$,
which proves the lemma.
\end{proof}

Now, define a subframe $I'$ to be {\em good} if $p_S \le \rho_{red}$
throughout $I'$, and otherwise $I'$ is called {\em bad}. With the help of
Lemma~\ref{fast:recovery} and Lemma~\ref{sector} we can prove the following
lemma.

\begin{lemma}\label{lem:upper}
For any sector $S$, the expected number of bad subframes $I'$ in $I$ is at
most $1/polylog(N)$, and at most $\epsilon \beta'/6$ of the subframes $I'$ in
$I$ are bad w.h.p., where the constant $\beta'>0$ can be made arbitrarily small
depending on the constant $\delta$ in $f$. The bounds hold irrespective of the
events outside of $S$.
\end{lemma}

The proof can be found in \cite{disc10}.
Since we have exactly 6 sectors, it follows from Lemma~\ref{lem:upper} that
apart from an $\epsilon \beta'$-fraction of the subframes, all subframes $I'$
in $I$ satisfy $\sum_{v \in D_1(u)} p_v \leq 6\rho$ throughout $I'$ w.h.p.

\subsection{Zone 3}

Next, we consider Zone~3. We will show that although the aggregate probability
of the nodes in Zone~3 may be high (for some distributions of nodes in the
space it can actually be as high as $\Omega(n)$), their influence (or noise)
at node $v$ is limited if the radius of Zone 2 is sufficiently large. Thus,
probabilities recover quickly in Zone~1 and there are many opportunities for
successful receptions.

In order to bound the interference from Zone 3, we divide Zone 3 into two
sub-zones: $Z_3^-$, which contains all nodes from Zone 3 up to a radius of
$\bigO {\log^2 n}$, and $Z_3^+$, which contains all remaining nodes in Zone 3.
For Zone $Z_3^-$ we can prove the following lemma.

\begin{lemma}\label{lem:zone3-}
At most an $\epsilon \beta$-fraction of the subframes $I'$ in $I$ are bad for
some $R_1$-disk in Zone $Z_3^-$ w.h.p., where the constant $\beta'>0$ can be
made arbitrarily small depending on the constant $\delta$ in $f$.
\end{lemma}
\begin{proof}
The claim follows from the fact that the radius of Zone $Z_3^-$ is $\bigO
{\log^2 n}$ and hence $d=\bigO {\log^4 n}$ disks of radius $R_1$ are
sufficient to cover the entire area of $Z_3^-$. According to
Lemma~\ref{lem:upper}, over all of these disks, the expected number of bad
subframes is at most $1/polylog(N)$. Using similar techniques as for the proof
of Lemma~\ref{lem:upper} in \cite{disc10}, it can also be shown that for each
disk $D$, the probability for $D$ to have $k$ bad subframes is at most
$1/polylog(N)^k$ irrespective of the events outside of $D$. Hence, one can use
Chernoff bounds for sums of identically distributed geometric random variables
to conclude that apart from an $\epsilon \beta'/d$-fraction of the subframes,
all subframes $I'$ in $I$ satisfy $\sum_{v \in D} p_v \leq 6\rho$ throughout
$I'$ w.h.p. This directly implies the lemma.
\end{proof}

Suppose that $R_2 = c \cdot R_1$. Lemma~\ref{lem:zone3-} implies that in a
good subframe the expected noise level at any node $w \in D_1(v)$ created by
transmissions in Zone $Z_3^-$ is upper bounded by
\[
  6 \rho_{red} \cdot \sum_{d = (c-1)}^{\bigO{\log^2 n}} \frac{2\pi (d+1)}{\sqrt{2} (d
  R_1)^{\alpha}}
  \le \frac{12 \pi \rho_{red}}{\alpha-1} \cdot \frac{1}{(c-2)^{\alpha-2}
  R_1^{\alpha}}
\]
which is at most $\epsilon \vartheta/4$ if $c=O((1/\epsilon)^{1/(\alpha-2)})$
is sufficiently large. In order to bound the noise level at any node $w \in
D_1(v)$ from Zone $Z_3^+$, we prove the following claim.

\begin{claim}
Consider some fixed $R_1$-disk $D$. If at the beginning of time frame $I$,
$T_w \le \sqrt{F}$ for all $w \in D$, then for all time steps except for the
first subframe in $I$, $p_D \in \bigO {\log n}$, w.h.p.
\end{claim}
\begin{proof}
Lemma~\ref{fast:recovery} implies that there must be a time step $t$ in the
first subframe of $I$ with $p_D \le 6 \rho_{green}$ w.h.p. Since for $p_D \in
\Omega(\log{n})$ at least a logarithmic number of nodes in $D$ transmit and
therefore every node sees a busy channel, w.h.p., and $p_D$ can only increase
if a node sees an idle channel, $p_D$ is bounded by $\bigO{\log n}$ for the
rest of $I$ w.h.p.
\end{proof}

The claim immediately implies the following result.

\begin{lemma} \label{lem:zone3+}
If at the beginning of time frame $I$, $T_w \le \sqrt{F}$ for all $w$, then
for all time steps except for the first subframe in $I$, the interference at
any node $w \in D_1(v)$ due to transmissions in $Z_3^+$ is at most $\epsilon
\vartheta/4$ w.h.p.
\end{lemma}

Hence, we get:

\begin{lemma} \label{lem:zone3}
If at the beginning of time frame $I$, $T_w \le \sqrt{F}$ for all $w$, then at
most an $\epsilon \beta$-fraction of the subframes in $I$ contain time steps
in which the expected interference at any node $w \in D_1(v)$ due to
transmissions in Zone 3 is at least $\epsilon \vartheta/2$.
\end{lemma}

\subsection{Zone 2}

For Zone $Z_2$ we can prove the following lemma in the same way as
Lemma~\ref{lem:zone3-}.

\begin{lemma}\label{lem:zone2}
At most an $\epsilon \beta$-fraction of the subframes $I'$ in $I$ are bad for
some $R_1$-disk in Zone 2, w.h.p., where the constant $\beta>0$ can be made
arbitrarily small depending on the constant $\delta$ in $f$.
\end{lemma}

\subsection{Throughput}

Given the upper bounds on the aggregate probabilities and interference, we
are now ready to study the throughput of $\ALG$. For this we first need to
show an upper bound on $T_v$ in order to avoid long periods of high $p_v$
values. Let $J$ be a time interval that has a quarter of the length of a time
frame, i.e., $|J|=F/4$. We start with the following lemma whose proof is
identical to Lemma III.6 in \cite{icdcs11stefan}.

\begin{lemma} \label{cl_Tinc}
If in subframe $I'$ the number of idle time steps at $v$ is at most $k$, then
node $v$ increases $T_v$ by 2 at most $k/2+\sqrt{f}$ many times in $I'$.
\end{lemma}

Next, we show the following lemma.

\begin{lemma} \label{lem:tvbound}
If at the beginning of $J$, $T_v \le \sqrt{F}/2$ for all nodes $v$, then every
node $v$ has at least $2^{-\bigO{(1/\epsilon)^{2/(\alpha-2)}}} |J|$
time steps in $J$ in which it senses an idle channel, w.h.p.
\end{lemma}
\begin{proof}
Fix some node $v$. Let us call a subframe $I'$ in $J$ {\em good} if in Zone 1
and in any $R_1$-disk in Zone 2 of $v$, the aggregate probability is upper
bounded by a constant, and the expected interference due to transmissions at
$v$ induced from Zone 3 is at most $\epsilon \vartheta/2$ throughout $I'$.
From Lemmas~\ref{lem:upper}, \ref{lem:zone2}, and \ref{lem:zone3} it follows
that there is an $(1-\epsilon)$-fraction of good subframes in $J$. Since $R_2
= \bigO{(1/\epsilon)^{1/(\alpha-2)} R_1}$, for any time step $t$ in a good
subframe $I'$ the total aggregate probability in Zones 1 and 2 of $v$ is upper
bounded by $\bigO{(1/\epsilon)^{2/(\alpha-2)}}$. Hence, the probability that
none of the nodes in Zones 1 and 2 of $v$ transmits is given by
\[
  \sum_{w \in Z_1 \cup Z_2} (1-p_w) \ge e^{-2 \sum_{w \in Z_1 \cup Z_2} p_w}
  = 2^{-\bigO{(1/\epsilon)^{2/(\alpha-2)}}}
\]
Due to the Markov inequality, the probability that the interference due to
transmissions in Zone 3 is at least $\epsilon \vartheta$ is at most $1/2$.
These probability bounds hold independently of the other time steps in $I'$.
Moreover, the total interference energy of the adversary in $I'$ is bounded by
$|I'|(1-\epsilon)^2 \vartheta$, which implies that at most a
$(1-\epsilon)$-fraction of the time steps in $I'$ are potentially busy, i.e.,
${\cal ADV}(v) \ge (1-\epsilon)\vartheta$. Hence, for at least a
$2^{-\bigO{(1/\epsilon)^{2/(\alpha-2)}}}$-fraction of the time steps in $I'$,
the probability for $v$ to sense an idle channel is a constant, which implies
the lemma.
\end{proof}

This allows us to prove the following lemma.

\begin{lemma}
If at the beginning of $J$, $T_v \le \sqrt{F}/2$ for all $v$, then also $T_v
\le \sqrt{F}/2$ for all $v$ at the end of $J$, w.h.p.
\end{lemma}
\begin{proof}
From the previous lemma we know that every node $v$ senses an idle channel for
$\Omega(|J|)$ time steps in $J$ for any constants $\epsilon>0$ and $\alpha>2$.
$T_v$ is maximized at the end of $J$ if all of these idle time steps happen at
the beginning of $J$, which would get $T_v$ down to 1 at some point.
Afterwards, $T_v$ can rise to a value of at most $t$ for the maximum $t$ with
$\sum_{i=1}^t 2i \le |J|$ (because $v$ increases $T_v$ by 2 each time it sees
no idle channel in the previous $T_v$ steps), which is at most $\sqrt{|J|}$.
Since $\sqrt{|J|} = \sqrt{|F|}/2$, the lemma follows.
\end{proof}

Since $T_v$ can be increased at most $(F/4)\sqrt{F}/2$ many times in $J$, we
get:

\begin{lemma}
If at the beginning of a time frame $I$, $T_v \le \sqrt{F}/2$ for all $v$,
then throughout $I$, $T_v \le \sqrt{F}$ for all $v$, and at the end of $I$,
$T_v \le \sqrt{F}/2$ for all $v$, w.h.p.
\end{lemma}

Hence, the upper bounds on $T_v$ that we assumed earlier are valid w.h.p. We
are now ready to prove Theorem~\ref{th:main}.

\begin{proof}[of Theorem~\ref{th:main}]
Recall that a time step is {\em open} for a node $v$ if $v$ and at least one
other node in $D_1(v)$ are not potentially busy. Let $J$ be the set of all
open time steps in $I$. Furthermore, let $k_0$ be the number of times $v$
senses an idle channel in $J$ and let $k_1$ be the number of times $v$
receives a message in $I$. From Lemma~\ref{lem:tvbound} and the assumptions in
Theorem~\ref{th:main} we know that $k_0 =
2^{-\bigO{(1/\epsilon)^{2/(\alpha-2)}}} |I|$.

\medskip

\noindent {\em Case 1:} $k_1 \ge k_0/6$. Then our protocol is
$2^{-\bigO{(1/\epsilon)^{2/(\alpha-2)}}}$-competitive for $v$ and we are done.

\medskip

\noindent {\em Case 2:} $k_1 < k_0/6$. Then we know from Lemma~\ref{cl_Tinc}
that $p_v$ is decreased at most $k_0/2 + \sqrt{F}$ times in $I$ due to $c_u >
T_u$. In addition to this, $p_v$ is decreased at most $k_1$ times in $I$ due
to a received message. On the other hand, $p_v$ is increased at least $k_0$
times in $J$ (if possible) due to an idle channel w.h.p. Also, we know from
our protocol that at the beginning of $I$, $p_v = \hat{p}$. Hence, there must
be at least $(1-1/2-1/6)k_0 - \sqrt{|F|} \ge k_0/4$ rounds in $J$ w.h.p. at
which $p_v = \hat{p}$. Now, recall the definition of a good subframe in the
proof of Lemma~\ref{lem:tvbound}. From Lemmas~\ref{lem:upper},
\ref{lem:zone2}, and \ref{lem:zone3} it follows that at most a $\epsilon
\beta$-fraction of the subframes in $I$ is bad. In the worst case, all of the
time steps in these subframes are open time steps, which sums up to at most
$k_0/8$ if $\beta$ is sufficiently small. Hence, there are at least $k_0/8$
rounds in $J$ that are in good subframes, w.h.p., and at which $p_v =
\hat{p}$, which implies that the other not potentially busy node in $D_1(v)$
has a constant probability of receiving a message from $v$. Using Chernoff
bounds, at least $k_0/16$ rounds with successfully received transmissions can
be identified for $v$, w.h.p.

If we charge $1/2$ of each successfully transmitted message to the sender and
$1/2$ to the receiver, then a constant competitive throughput can be
identified for every node in both cases above. It follows that our protocol is
$2^{-\bigO{(1/\epsilon)^{2/(\alpha-2)}}}$-competitive in $F$.
\end{proof}

Now, let us consider the two cases of Theorem~\ref{thm:main}. Recall that we
allow here any $((1-\epsilon)\vartheta, T)$-bounded adversary.

\begin{proof}[of Theorem \ref{thm:main}]
\subsubsection*{\\Case 1: the adversary is 1-uniform and $\forall v : D_1(v) \not = \emptyset$.}

In this case, every node has a non-empty neighborhood and therefore {\em all}
non-jammed rounds of the nodes are open. Hence, the conditions on a weakly
$((1-\epsilon)\vartheta, T)$-bounded adversary are satisfied. So
Theorem~\ref{th:main} applies, which completes the proof of
Theorem~\ref{thm:main} a).

\subsubsection*{Case 2: {\boldmath $|D_1(v)| \ge 2/\epsilon$} for all {\boldmath $v \in V$}.}

Consider some fixed time interval $I$ with $|I|$ being a multiple of $T$. For
every node $v \in D_1(u)$ let $f_v$ be the number of non-jammed rounds at $v$
in $I$ and $o_v$ be the number of open rounds at $v$ in $I$. Let $J$ be the
set of rounds in $I$ with at most one non-jammed node. Suppose that $|J| >
(1-\epsilon/2)|I|$. Then every node in $D_1(u)$ must have more than
$(\epsilon/2)|I|$ of its non-jammed rounds in $J$. As these non-jammed rounds
must be serialized in $J$ to satisfy our requirement on $J$, it holds that
$|J| > \sum_{v \in D_1(u)} (\epsilon/2)|I| \ge (2/\epsilon) \cdot
(\epsilon/2)|I| = |I|$. Since this is impossible, it must hold that $|J| \le
(1-\epsilon/2)|I|$.

Thus, $\sum_{v \in D_1(u)} o_v \ge (\sum_{v \in D_1(u)} f_v) - |J| \ge (1/2)
\sum_{v \in D_1(u)} f_v$ because $\sum_{v \in D_1(u)} f_v \ge (2/\epsilon)
\cdot \epsilon |I| =2|I|$. Let $D'(u)$ be the set of nodes $v \in D_1(u)$ with
$o_v \ge f_v/4$. That is, for each of these nodes, a constant fraction of the
non-jammed time steps is open. Then $\sum_{v \in D_1(u) \setminus D'(u)} o_v <
(1/4)\sum_{v \in D_1(u)} f_v$, so $\sum_{v \in D'(u)} o_v \ge (1/2) \sum_{v
\in D_1(u)} o_v\geq (1/4)\sum_{v \in D_1(u)} f_v$.

Consider now a set $U \subseteq V$ of nodes so that $\bigcup_{u \in U} D_1(u)
= V$ and for every $v \in V$ there are at most 6 nodes $u \in U$ with $v \in
D_1(u)$. Note $U$ is easy to construct in a greedy fashion for arbitrary UDGs,
and therefore for $D_1(u)$ in the SINR model, and also known as a {\em
dominating set of constant density}. Let $V'=\bigcup_{u \in U} D'(u)$. Since
$\sum_{v \in D'(u)} o_v \ge (1/4) \sum_{v \in D_1(u)} f_v$ for every node $u
\in U$, it follows that $\sum_{v \in V'} o_v \ge (1/6) \sum_{u\in U} \sum_{v
\in D'(u)} o_v \ge (1/24) \sum_{u\in U} \sum_{v \in D_1(u)} f_v \ge (1/24)
\sum_{v \in V} f_v$. Using that together with Theorem~\ref{th:main}, which
implies that {\sc Sade} is constant competitive w.r.t. the nodes in $V'$,
completes the proof of Theorem~\ref{thm:main} b).
\end{proof}

Finally, we show that {\sc Sade} is self-stabilizing, i.e., it can recover
quickly from any set of $p_v$- and $T_v$-values.

\subsection{Optimality}\label{sec:lower}

Obviously, if a jammer has a sufficiently high energy budget, it can essentially block all nodes all the time. In the following we call a network \emph{dense} if $\forall v \in V : D_1(v) \geq 1$.

\begin{theorem}\label{thm:main_nonuniform}
The nodes can be positioned so that the transmission range of every node is
non-empty and yet no MAC protocol can achieve any throughput against a
$(B,T)$-bounded adversary with $B > \vartheta$, even if it is 1-uniform.
\end{theorem}
\begin{proof}
Let us suppose the jammer uses an energy budget $B > \vartheta$. If every node $v$ only has nodes right at the border of its disk $D_1(v)$ and the adversary continuously sets ${\cal ADV}(v)=B$, then $v$ will not be able to receive any messages according to the SINR model. Thus the overall throughput in the system is 0.
\end{proof}

\section{Simulations} \label{sec:simulation}

To complement our formal analysis and to investigate the average-case behavior of our protocol, we conducted a simulation study.
In the following, we consider two scenarios which differ in the way nodes are distributed in the 2-dimensional Euclidean space.
In the first scenario, called \textsc{Uni}, the nodes are distributed \emph{uniformly at random} in the 2-dimensional plane of size $25\times 25$ units. In the second scenario, called \textsc{Het}, we first subdivide the 2-dimensional plane of size $25 \times 25$ units into 25 \emph{sub-squares} of size $5\times 5$ units. For each sub-square we then choose the number of nodes $\lambda$ uniformly at random from the interval $[20,1000]$ and distribute said nodes (uniformly at random) in the corresponding sub-square. Consequently, each sub-square potentially provides a different density, where the attribute density represents the average amount of nodes on a spot in the plane of the corresponding scenario. In order to avoid boundary effects, for both \textsc{Uni} and \textsc{Het}, we assume that the Euclidean plane ``wraps around'', i.e., distances are computed modulo the boundaries.

While our formal throughput results in Section~\ref{sec:analysis} hold for \emph{any} adversary which respects the jamming budget constraints, computing the best adversarial strategy (i.e., the strategy which minimizes the throughput of $\ALG$) is difficult. Hence, in our simulations, we consider the following two types of adversaries:
(1) \emph{Regular (or random) jammer (\textsc{Reg}):} given an energy budget $B$ per node, a time interval $T$, and a specific $1 > \epsilon > 0$, the adversary randomly jams each node every $\epsilon$th round (on average) using exactly $\frac{B}{\epsilon}$ energy per node. Additionally we make sure that the overall budget $B$ is perfectly used up at the end of $T$.
(2) \emph{Bursty (or deterministic) jammer (\textsc{Bur}):} For each time period $T$, the adversary jams \emph{all} initial rounds at the node, until the budget $B$ is used up. The remaining rounds in $T$ are unjammed. In other words, the first $\epsilon T$ many rounds are jammed by the adversary using exactly $\frac{B}{\epsilon}$ energy per node.

If not stated otherwise, we use the jammer \textsc{Reg} and parameters $\alpha=3, \epsilon = \frac{1}{3}$, $\beta=2$, $\Pi=8$, $T=60$, $B=(1-\epsilon)\cdot \beta$ and run the experiment for 3000 rounds. We will typically plot the percentage of successful message receptions, averaged over all nodes, with respect to the \emph{unjammed time steps}. If not specified otherwise, we repeat each experiment ten times with different random seeds, both for the distribution of nodes in the plane as well as the decisions made by our MAC protocol. By default, our results show the \emph{average} over these runs; the variance of the runs is low.

\textbf{Impact of Scale and $\alpha$.} We first study the throughput as a function of the network size. Therefore we distribute $n$ nodes uniformly in the $\sqrt{n}\times \sqrt{n}$ plane. Figure~\ref{fig:impact_scale} (\emph{top left}) shows our results under the \textsc{Reg} (or \emph{random}) jammer and different $\alpha$ values.
\begin{figure*} [t]
\begin{center}
\includegraphics[width=.38\textwidth]{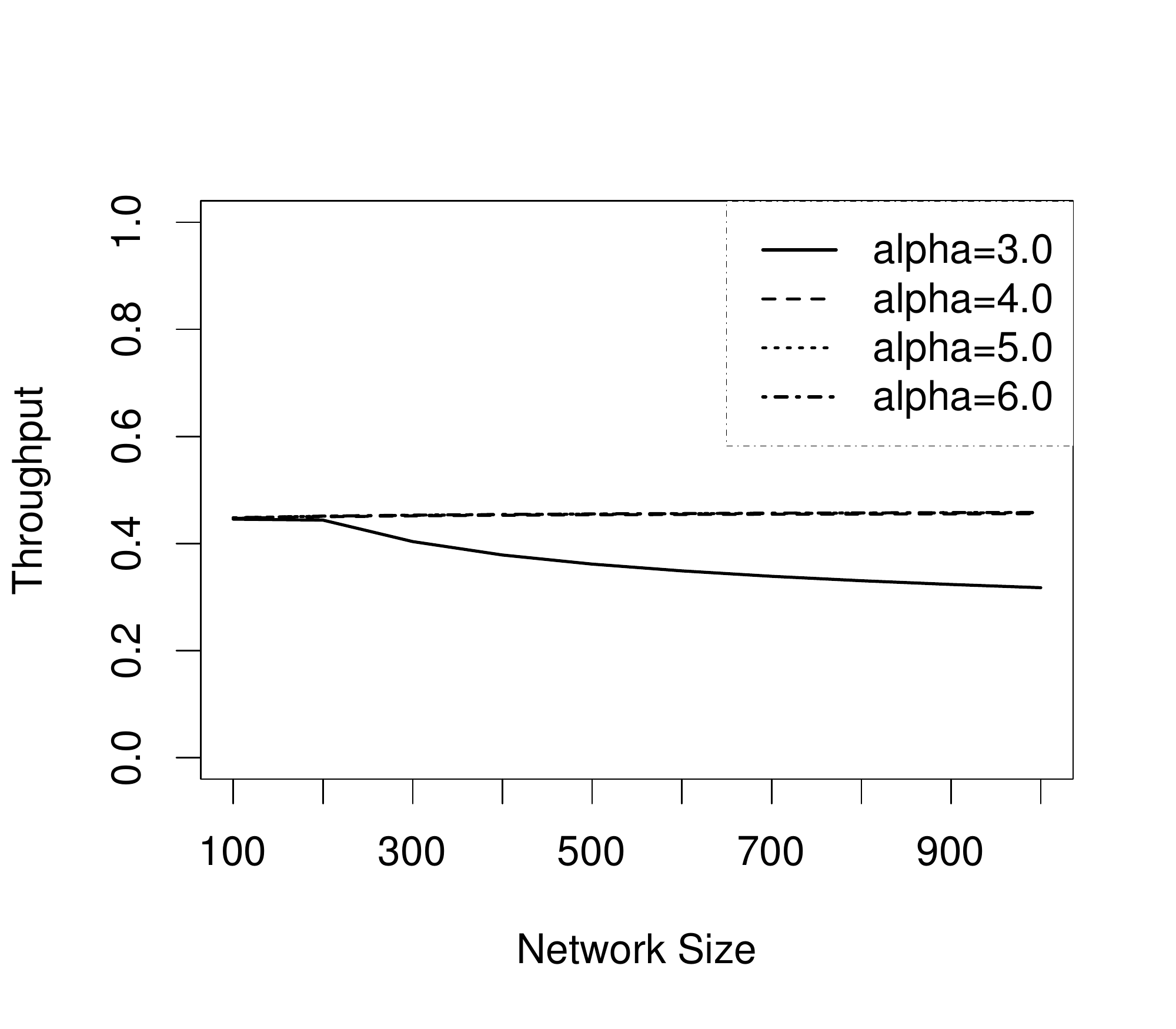}~~~~~~~~\includegraphics[width=.38\textwidth]{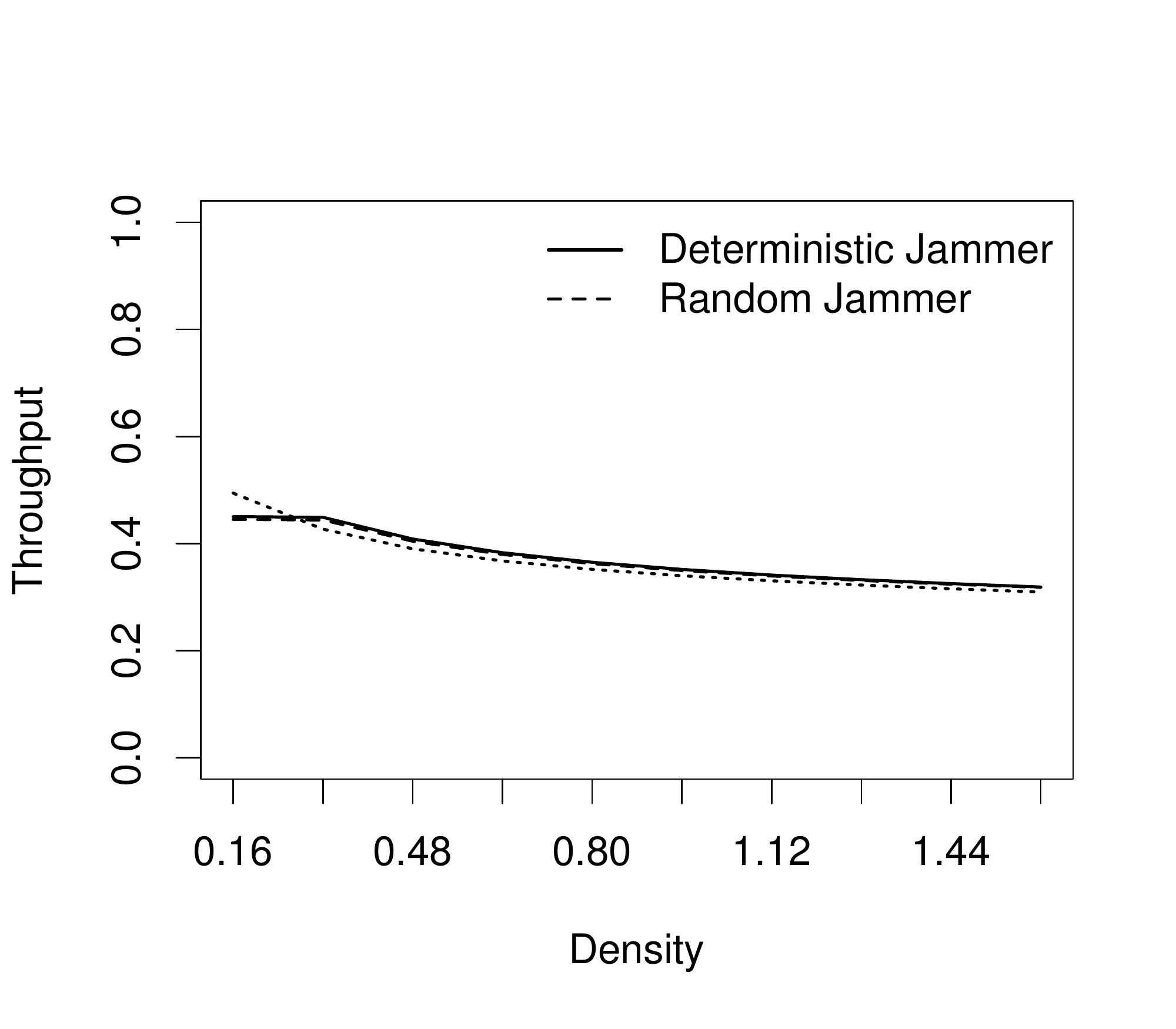}\\
\includegraphics[width=.38\textwidth]{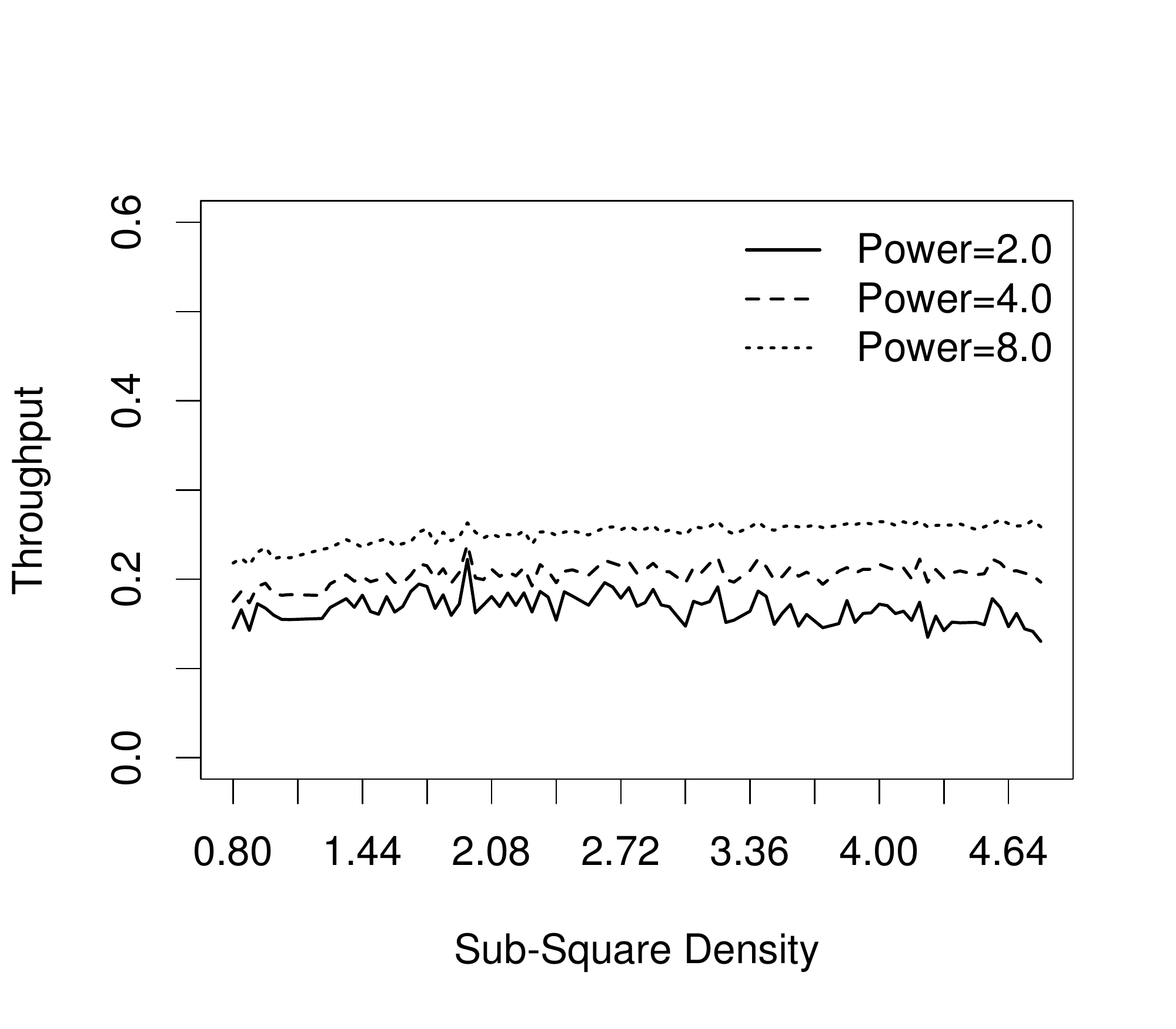}~~~~~~~~\includegraphics[width=.38\textwidth]{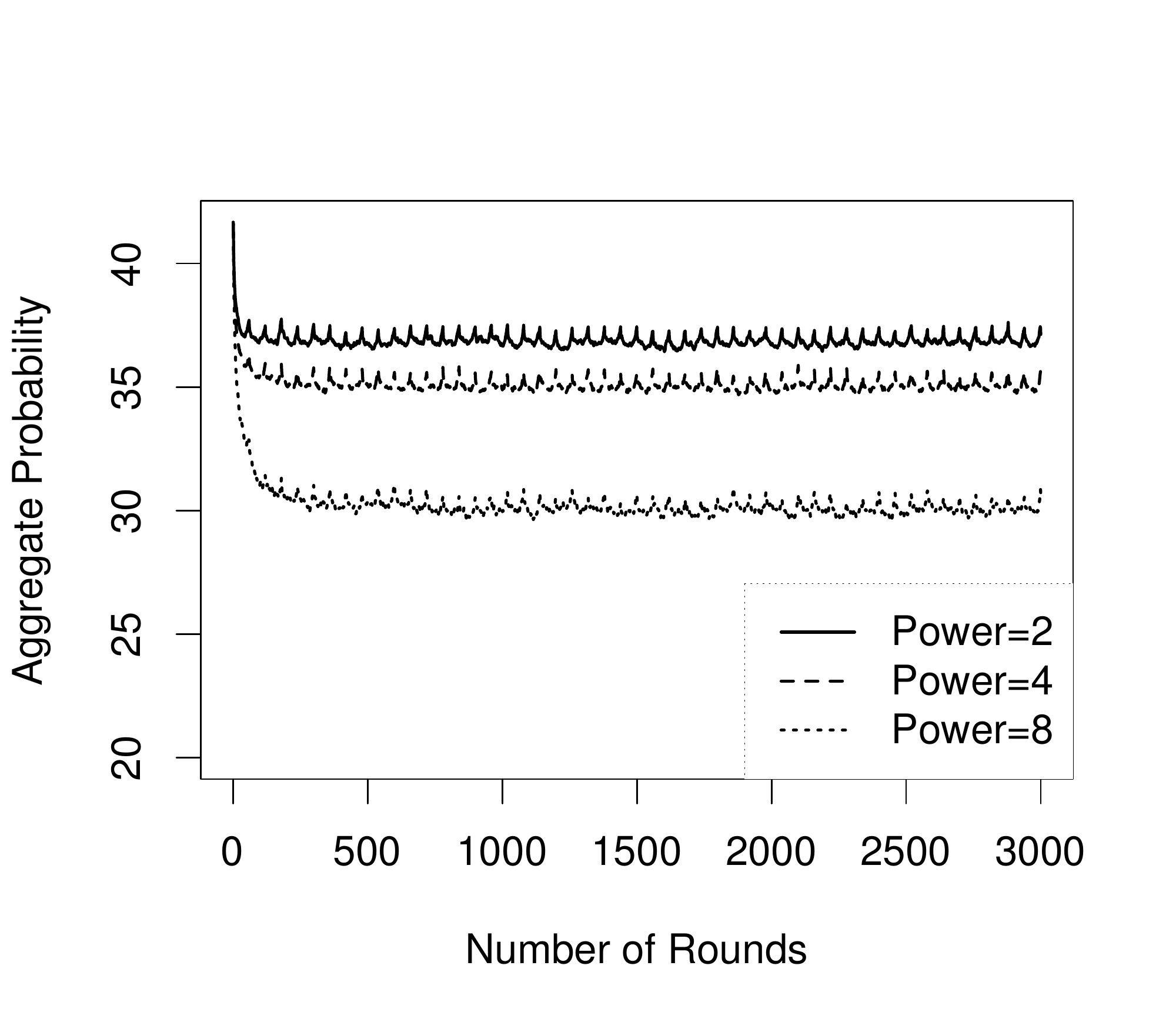}
\caption{Throughput as a function of network size (\emph{top left}), density (\emph{top right}), sub-grid density (\emph{bottom left}), and power (\emph{bottom right}).}\label{fig:impact_scale}
\end{center}
\end{figure*}
First, we can see that the competitive throughput is around 40\%,
which is higher than what we expect from our worst-case formal analysis.
Interestingly, for $\alpha=3$, we observe a small throughput decrease for larger networks;
but for $\alpha>3$, the throughput is almost independent of the network scale. (In the literature, $\alpha$ is typically modeled as 3 or 4.)

This partially confirms Theorem~\ref{thm:main}: a higher $\alpha$ renders the transmissions and power propagation more local. This locality can be exploited by $\ALG$ to some extent.

\textbf{Impact of Density.} Next, we investigate how the performance of $\ALG$ depends on the node density. We focus on $\alpha=3$
and study both the \textsc{Reg} jammer as well as the \textsc{Bur} (deterministic) jammer.
Figure~\ref{fig:impact_scale} (\emph{top right}) shows that results for the \textsc{Uni} scenario ($n$ nodes distributed uniformly in the $25\times 25$ plane, i.e., density $n/625$). The throughput is similar under both jammers, and slightly declines for denser networks.
This effect is very similar to the effect of having larger (but equidistant) networks.


However, $\ALG$ suffers more from more heterogenous densities.
 The results for the scenario \textsc{Het} are shown in Figure~\ref{fig:impact_scale} (\emph{bottom left}).
 While the throughput is generally lower, the specific sub-square density plays a minor role.



\textbf{Convergence Time.}
$\ALG$ adapts quite fast to the given setting, as the nodes increase and decrease their sending probabilities in a multiplicative manner. Being able to adapt quickly is an important feature, in particular in dynamic or mobile environments where nodes can join and leave over time, or where nodes are initialized with too high or low sending probabilities. Our distributed MAC protocol will adjust automatically and ``self-stabilize''.

Figure~\ref{fig:impact_scale} (\emph{bottom right}) shows representative executions over time and plots the aggregate probability.
Initially, nodes have a maximum sending probability $\hat{p}=1/24$. This will initially lead to many collisions; however, very quickly, the senders back off and the overall sending probabilities (the \emph{aggregated probability}) reduce almost exponentially, and we start observing successful message transmissions. (Observe that the aggregated ``probability'' can be higher than one, as it is simply the sum of the probabilities of the individual nodes.)

The sum of all sending probabilities also converges quickly for any other $\Pi$. However, for smaller powers, the overall probability is higher. This is consistent with the goal of $\ALG$: because for very large sending powers, also more remote nodes in the network will influence each other and interfere, it is important that there is only a small number of concurrent senders in the network at any time---the aggregated sending probability must be small. On the other hand, small powers allow for more local transmissions, and to achieve a high overall throughput, many senders should be active at the same time---the overall sending probability should be high.


\textbf{802.11a and Impact of Epsilon.}
We also compared the throughput of $\ALG$ to the standard 802.11 MAC protocol (with a focus on 802.11a). For simplicity, we set the unit slot time for 802.11 to 50~$\mu s$. The backoff timer of the 802.11 MAC protocol implemented here uses units of 50~$\mu s$. We omit SIFS, DIFS, and RTS/CTS/ACK. Our results show that 802.11a suffers more from the interference, while it yields a similar throughput for large $\epsilon$. In fact, we find that for $\epsilon$ close to 0, 802.11a can even slightly outperform $\ALG$.

When varying $\epsilon$, we find that the worst-case bound of Theorem~\ref{thm:main} may be too pessimistic in many scenarios, and the throughput depends to a lesser extent on the constant $\epsilon$.

\section{Related Work}\label{sec:relwork}


Traditional jamming defense mechanisms typically operate on the physical layer
\cite{LiuNST07,NavdaBGR07,SimonOSL01}, and mechanisms have been designed to both {\em
avoid} jamming as well as {\em detect} jamming. Especially spread spectrum technology is very effective to avoid jamming, as with widely spread
signals, it becomes harder to detect the start of a packet quickly enough in
order to jam it. Unfortunately, protocols such as IEEE 802.11b use relatively
narrow spreading~\cite{IEEE99}, and some other IEEE 802.11 variants spread
signals by even smaller factors~\cite{Brown:mobihoc06}. Therefore, a jammer
that simultaneously blocks a small number of frequencies renders spread
spectrum techniques useless in this case. As jamming strategies can come in
many different flavors, detecting jamming activities by simple methods based
on signal strength, carrier sensing, or packet delivery ratios has turned out
to be quite difficult~\cite{Li:infocom07}.

Recent work has investigated \emph{MAC layer strategies} against jamming in more detail,
for example coding strategies~\cite{Chiang:mobicom07}, channel surfing and
spatial retreat~\cite{Alnifie:Q2SWinet07,Xu:wws04}, or mechanisms to hide
messages from a jammer, evade its search, and reduce the impact of corrupted
messages~\cite{Wood:secon07}. Unfortunately, these methods do not help against
an adaptive jammer with {\em full} information about the history of the
protocol, like the one considered in our work.

In the theory community, work on MAC protocols has mostly focused on
efficiency. Many of these protocols are random backoff or tournament-based
protocols~\cite{Bender05,Chlebus06,Gold00,Hastad96,Kwak05,Raghavan99} that do
not take jamming activity into account and, in fact, are not robust against it
(see~\cite{singlehop08} for more details). The same also holds for many MAC
protocols that have been designed in the context of broadcasting~\cite{CR06}
and clustering~\cite{Kuhn04}.
Also some work on jamming is known (e.g.,~\cite{citer2} for a short overview).
There are two basic approaches in the literature. The first assumes randomly
corrupted messages (e.g.~\cite{PP05}), which is much easier to handle than
adaptive adversarial jamming~\cite{Raj08}. The second line of work either
bounds the number of messages that the adversary can transmit or disrupt with
a limited energy budget (e.g.~\cite{GGN06,KBKV06}) or bounds the number of
channels the adversary can jam
(e.g.~\cite{dolevpodc,DGGN07,DGGN08,shlomi07,GGKN09,seth09,dcoss09}).

The protocols in~\cite{GGN06,KBKV06} can tackle adversarial jamming at both
the MAC and network layers, where the adversary may not only be jamming the
channel but also introducing malicious (fake) messages (possibly with address
spoofing). However, they depend on the fact that the adversarial jamming
budget is finite, so it is not clear whether the protocols would work under
heavy continuous jamming. (The result in~\cite{GGN06} seems to imply that a
jamming rate of $1/2$ is the limit whereas the handshaking mechanisms in
\cite{KBKV06} seem to require an even lower jamming rate.)

In the multi-channel version of the problem introduced in the theory community
by Dolev~\cite{shlomi07} and also studied in
\cite{dolevpodc,DGGN07,DGGN08,shlomi07,GGKN09,seth09,dcoss09}, a node can only
access one channel at a time, which results in protocols with a fairly large
runtime (which can be exponential for deterministic protocols
\cite{DGGN07,GGKN09} and at least quadratic in the number of jammed channels
for randomized protocols~\cite{DGGN08,dcoss09} if the adversary can jam almost
all channels at a time). Recent work~\cite{dolevpodc} also focuses on the
wireless synchronization problem which requires devices to be activated at
different times on a congested single-hop radio network to synchronize their
round numbering while an adversary can disrupt a certain number of frequencies
per round. Gilbert et al.~\cite{seth09} study robust information exchange in
single-hop networks.

Our work is motivated by the work in~\cite{Raj08} and~\cite{singlehop08}. In
\cite{Raj08} it is shown that an adaptive jammer can dramatically reduce the
throughput of the standard MAC protocol used in IEEE 802.11 with only limited
energy cost on the adversary side. Awerbuch et al.~\cite{singlehop08}
initiated the study of throughput-competitive MAC protocols under continuously
running, adaptive jammers, but they only consider single-hop wireless
networks. Their approach has later been extended to reactive jamming environments~\cite{icdcs11stefan}, co-existing networks~\cite{podc12stefan}
and applications such as leader election~\cite{mobihoc11stefan}.

The result closest to ours is the robust MAC protocol for Unit Disk Graphs presented in~\cite{disc10}.
In contrast to~\cite{disc10}, we initiate the study of the more relevant and realistic \emph{physical interference model}~\cite{sinr-original}
and show that a competitive throughput can still be achieved.
As unlike in Unit Disk Graphs, in the SINR setting far-away communication
can potentially interfere and there is no absolute notion of an idle medium, a new protocol is needed whose
geometric properties must be understood. For the SINR setting, we also introduce a new adversarial model (namely the \emph{energy budget adversary}).


\section{Conclusion}\label{sec:conclusion}

This paper has shown that robust MAC protocols achieving a constant competitive throughput exist even in the physical model. This concludes a series of research works in this area.
 Nevertheless, several interesting questions remain open.
For example, while our theorems prove that  $\ALG$ is as robust
as a MAC protocol can get within our model and for constant $\epsilon$, we conjecture that a throughput
which is polynomial in
$(1/\epsilon)$ is possible. However, we believe that such a claim is very difficult to prove.
We also plan to explore the performance of $\ALG$ under specific node mobility patterns.

\textbf{Acknowledgments.} The authors would like to thank Michael Meier from Paderborn University for his
help with the evaluation of the protocol.

{\footnotesize \renewcommand{\baselinestretch}{.8}
  \bibliographystyle{abbrv} \bibliography{jammers}
}

\end{document}